\documentclass[
    ,final            % use final for the camera ready runs
  ]
  {aipproc}

\layoutstyle{8x11single}

\usepackage{amsmath}
\usepackage{amsfonts}
\usepackage{amssymb}
\usepackage{amsthm}
\usepackage{newlfont}
\usepackage{graphicx}
\usepackage{amscd}
\usepackage{accents}

\theoremstyle{plain}
\newtheorem{Th}{Theorem}%[section]

\newtheorem{Prop}[Th]{Proposition}
\theoremstyle{definition}
\newtheorem{Def}{Definition}%[section]
\theoremstyle{remark}
\newtheorem*{Rem}{Remark}%[section]
\newcommand{\EE}{{\mathbb E}}

\newcommand{\DD}{{\mathbb D}}
\newcommand{\FF}{{\mathbb F}}

\newcommand{\ZZ}{{\mathbb Z}}

\newcommand{\bphi}{\mathbf{\phi}}

\begin{document}

\title{Desargues maps and their reductions}


\author{Adam Doliwa}{
  address={Faculty of Mathematics and Computer Science, University of Warmia and Mazury in Olsztyn,
ul.~S{\l}oneczna~54, 10-710~Olsztyn, Poland 
%\\ \medskip
%Institute of Mathematics, Polish Academy of Sciences, ul.~\'{S}niadeckich 8, 00-956 Warszawa, Poland}
}
}

\begin{abstract}
We present recent developments on geometric theory of the Hirota system and of the non-commutative discrete Kadomtsev--Petviashvili (KP) hierarchy adding also some new results which make the picture more complete. We pay special attention to multidimensional consistency of the Desargues maps and of the resulting non-linear non-commutative systems. In particular, we show three-dimensional consistency of the non-commutative KP map in its edge formulation. We discuss also relation of Desargues maps and quadrilateral lattice maps. We study from that point of view reductions of the Hirota system to discrete $B$-KP and $C$-KP systems presenting also a novel constraint which leads to the Miwa equations. By imposing periodicity reduction of the discrete KP hierarchy we obtain non-isospectral versions of the modified lattice Gel'fand--Dikii equations. To close the picture from below, we apply additional self-similarity constraint on the non-isospectral non-autonomous modified lattice Korteweg--de~Vries system to recover known $q$-Painleve equation of type $A_2+A_1$.
\end{abstract}

\maketitle

%%%%%%%%%%%%%%%%%%%%%%%%%%%%%%%%%%%%%%%%%%%%
%% MAINMATTER
%%%%%%%%%%%%%%%%%%%%%%%%%%%%%%%%%%%%%%%%%%%%

\section{Introduction}
Discrete integrable systems play at present the key role in the whole integrability theory. According to Kruskal 
\cite{GrammaticosRamani-INS}: "For years we have been thinking that the integrable evolution equations were the fundamental ones. It is becoming clear now that the fundamental objects are the integrable \emph{discrete} equations." Among discrete integrable systems the discrete KP equation proposed by Hirota~\cite{Hirota} takes a particular position. It was originally designated as a discrete version of the two dimensional Toda system~\cite{Mikhailov}, but it turned out \cite{Miwa} that it has a profound relation with the whole KP hierarchy of integrable equations. Moreover, under the name of $T$- or $Y$-system it plays an important role in solvable lattice models of quantum physics and statistical mechanics~\cite{KNS-rev}; see also \cite{Zabrodin} for a review of various properties of its classical version. Recently, another impetus to study the Hirota equation came from combinatorics \cite{Knutson}, where it is known as the octahderon recurrence. 

Already in works Darboux~\cite{Darboux-OS,DarbouxIV} one can find a geometric meaning of two dimensional Toda system as  equations governing projective invariants of the so called Laplace transformations of conjugate nets on surfaces. A transition of the  geometric picture to a discrete level~\cite{DCN} produces again the Hirota system, and leads to consideration of multidimensional lattices of planar quadrilaterals~\cite{MQL}, which are discrete analogs of multidimensional conjugate nets. Quite recently, it turned out~\cite{Dol-Des} that the theory of such quadrilateral lattices, where the underlying geometric constraint is coplanarity of four points, forms a part of the theory of Desargues maps with the constraint being just collinearity of three points. Such property looks rather trivial, but when the combinatorics of the points is prescribed~\cite{Dol-AN} according to the structure of the $A$-type root lattice, it leads to the Hirota system in its non-commutative version~\cite{Nimmo-NCKP}, we remark that non-commutative (matrix valued) equations of Hirota type appeared earlier in~\cite{LeviPilloniSantini,FWN}. 

Non-commutative integrable systems~\cite{Kupershmidt} can be considered as the second extreme point, the first being the systems with commuting dependent variables, with the quantum integrable systems in the middle. We will not consider here quantum integrability properties of the Hirota equation, but we refer to above mentioned review~\cite{KNS-rev}, the paper \cite{KashaevReshetikhin} and our previous work~\cite{DoliwaSergeev-pentagon,Dol-WCR-Hirota}. For quantum systems related to quadrilateral lattice maps see~\cite{BaMaSe,Sergeev-q3w}.  

In the present paper, after presenting the Desargues maps and their relation with the (non-commutative) Hirota system and discrete KP hierarchy, we concentrate on connection of the above to theory of quadrilateral lattice maps. In partcular, we offer a point of view on the discrete $B$-KP and $C$-KP equations from that perspective. Then we study periodic reductions of the Desargues maps which provide geometric meaning to the lattice Gel'fand--Dikii systems. Our approach makes clear the appearance of arbitrary functions of single variables in the systems and allows to introduce one more function related to non-isospectrality of the corresponding linear problems. Finally, we conclude the reduction procedure by recovering known $q$-Painleve equation of type $A_2+A_1$.

\section{Desargues maps}
\label{sec:Desargues}
\subsection{Desargues maps and the non-commutative Hirota system}

Consider~\cite{ConwaySloane} the $\widehat{N}$-dimensional root lattice $Q(A_{\widehat{N}})$ 
as generated by vectors along the edges of regular $\widehat{N}$-simplex in the $\widehat{N}$-dimensional euclidean space 
$\EE^{\widehat{N}}$. 
The \emph{holes} of the lattice are the points of the ambient space that are
locally maximally distant from the lattice. The convex hull of the lattice
points closest to a hole is called the Delaunay polytope. Among the Delaunay polytopes of the root lattice $Q(A_{\widehat{N}})$
are the so called basic $\widehat{N}$-simplices, which are translates of the initial $\widehat{N}$-simplex. 

Denote by ${\mathbb P}^M({\mathbb D})$, $M$-dimensional right projective space over a division ring $\DD$.
\begin{Def}
Desargues maps are defined~\cite{Dol-Des,Dol-AN} as maps $\Phi:Q(A_{\widehat{N}})\to{\mathbb P}^M({\mathbb D})$, $M\geq 2$, such that vertices of each basic $\widehat{N}$-simplex are mapped into collinear points, see Figure~\ref{fig:Desargues-map}. 
\end{Def}
\begin{figure}
\includegraphics[width=6cm]{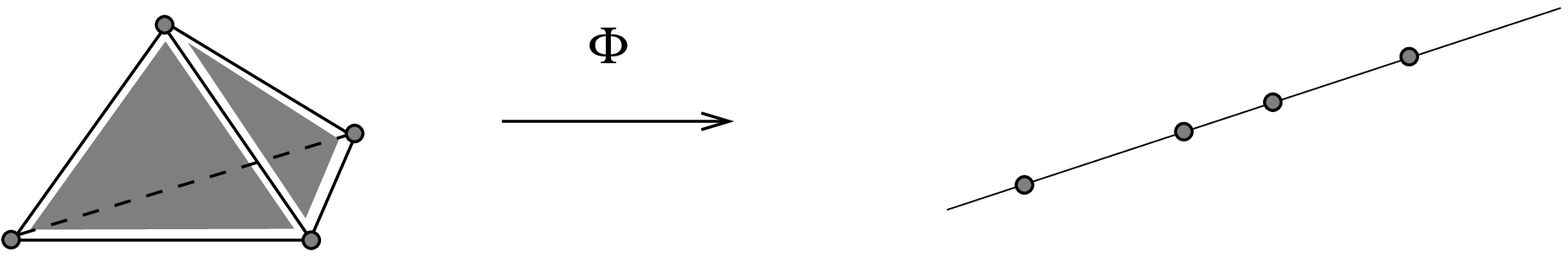}\hskip1.5cm
\includegraphics[width=7cm]{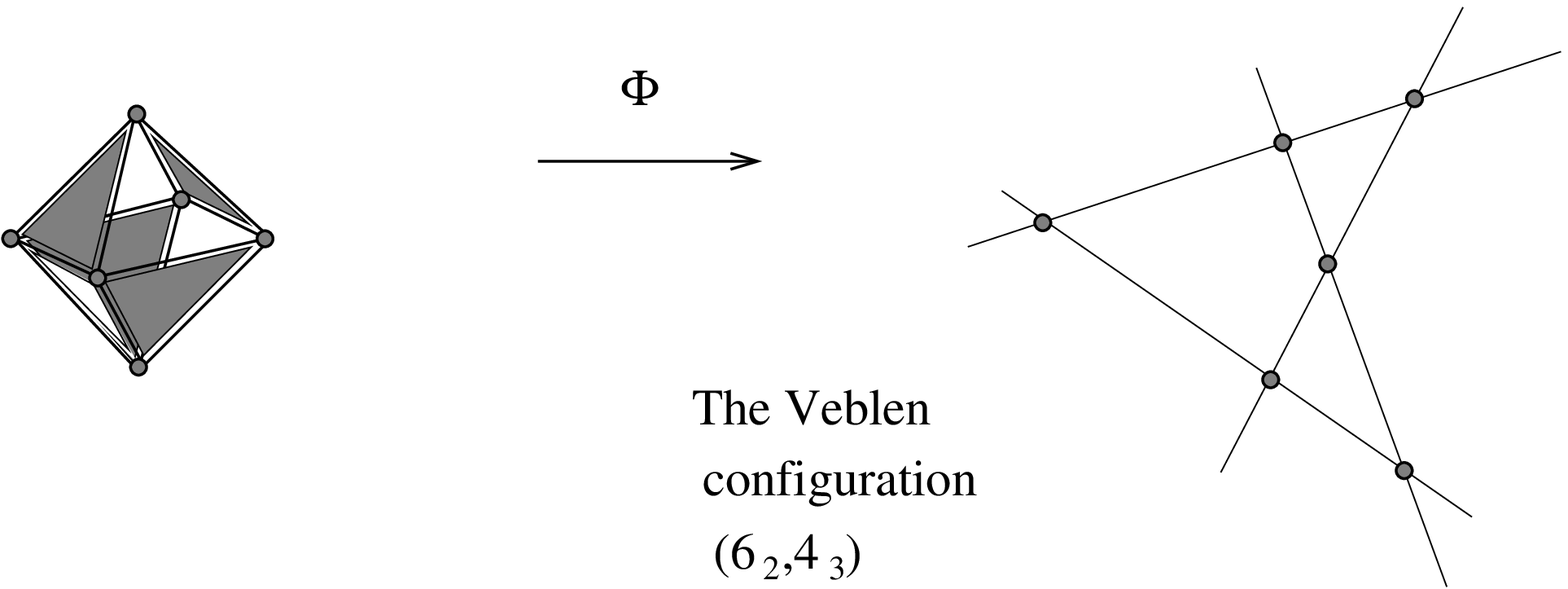}
\caption{Four collinear points as image of vertices of a basic $3$-simplex (left), and the Veblen configuration as the Desargues map image of the octahedron Delaunay polytope (black triangles belong to basic simplices)} 
\label{fig:Desargues-map}
\end{figure}

In the standard realization of $Q(A_{\widehat{N}})$ as sub-lattice of $\ZZ^{\widehat{N}+1}$ characterized by the condition $m_1 + m_2 + \dots + m_{\widehat{N}+1} = 0$, $(m_1,\dots ,m_{\widehat{N}+1})\in\ZZ^{\widehat{N}+1}$ fix the basis 
$\mathbf{\varepsilon}_i = \mathbf{e}_{\widehat{N}+1} - \mathbf{e}_i$, $i=1,\dots \widehat{N}$, of $Q(A_{\widehat{N}})$, where  $\mathbf{e}_j$, $i=1,\dots \widehat{N}+1$ are the elements of the standard basis of $\ZZ^{\widehat{N}+1}$.
After identification of $Q(A_{\widehat{N}})$ with $\ZZ^{\widehat{N}} = \sum_{i=1}^{\widehat{N}} \ZZ\mathbf{\varepsilon}_i$, in suitable gauge \cite{Dol-Des}  the homogeneous coordinates $\bphi:\ZZ^{\widehat{N}}\to\DD^{M+1}$ of Desargues maps satisfy the following linear system
\begin{equation} \label{eq:lin-dKP}
\mathbf{\phi}(\widehat{n}+\mathbf{\varepsilon}_i) - 
\mathbf{\phi}(\widehat{n}+\mathbf{\varepsilon}_j) =  
\mathbf{\phi}(\widehat{n}) 
U_{ij}(\widehat{n}),  \qquad 1\leq i \ne j \leq {\widehat{N}}, \qquad \widehat{n}=(n_1, n_2,\dots , n_{\widehat{N}}) \in\ZZ^{\widehat{N}},
\end{equation}
well known in soliton theory~\cite{DJM-II,Nimmo-NCKP}. In what follows we denote (forward and backward) shifts by (signed) subscripts in round brackets, and we skip the argument $\widehat{n}$, i.e. the linear system above reads
$\mathbf{\phi}_{(i)}- 
\mathbf{\phi}_{(j)} =  
\mathbf{\phi}
U_{ij}$.
Its compatibility condition
\begin{equation} \label{eq:comp-U}
U_{ij} + U_{ji} = 0,  \qquad  U_{ij} + U_{jl} + 
U_{li} = 0, \qquad
U_{lj}U_{li(j)} = 
U_{li} U_{lj(i)}, \qquad  i,j,l \qquad \text{distinct},
\end{equation}
is called the non-commutative Hirota system
\cite{FWN,Nimmo-NCKP}.
Geometrically, equations \eqref{eq:comp-U} describe the fact that vertices of any octahedron Delunay polytope of the root lattice are mapped to the so called Veblen configuration consisting of six points and four lines --- each line is incident with three points, and each point is incident with two lines (see Figure~\ref{fig:Desargues-map}). This point of view was advocated in \cite{Schief-talk} and motivated our definition of Desargues maps.

We remark that when $\DD$ is \emph{commutative} (i.e. a field, and then we write $\FF$ instead of $\DD$) then the functions $U_{ij}$ can be parametrized in terms of a single potential  (the tau-function)
$\tau \colon Q(A_{\widehat{N}})\to\FF$
\begin{equation} \label{eq:U-tau}
U_{ij} = \frac{\tau_{(ij)} \tau }{\tau_{(i)}\tau_{(j)}} = - U_{ji} , \qquad i<j,
\end{equation} 
and remaining equations
 reduce to the celebrated Hirota system \cite{Hirota}
\begin{equation} \label{eq:H-M}
\tau_{(i)}\tau_{(jl)} - \tau_{(j)}\tau_{(il)} + \tau_{(l)}\tau_{(ij)} =0,
\qquad 1\leq i< j < l \leq \widehat{N}.
\end{equation} 
The presented above definition of Desargues maps and the corresponding approach to the Hirota system via the root lattices 
$Q(A_{\widehat{N}})$ exhibits from the very beginning their invariance with respect to the affine Weyl group $W(A_{\widehat{N}})$ which acts on the root lattice (see~\cite{Dol-AN} for detailed discussion).
\begin{figure}
\includegraphics[width=8cm]{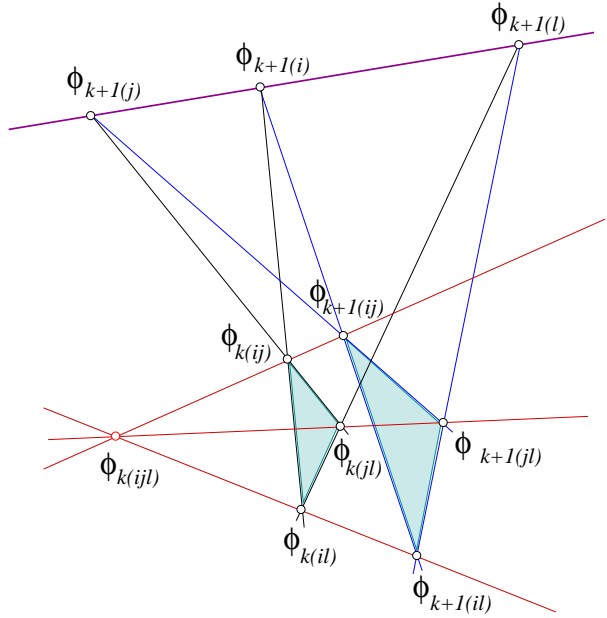}
\caption{Four dimensional consistency of Desargues maps} 
\label{fig:Desargues-phi-k}
\end{figure}

We remark, that after suitable gauge transformation one can derive a non-autonomous version \cite{WTS} of the Hirota system \eqref{eq:H-M} 
\begin{equation} \label{eq:H-M-na}
(\mathcal{A}_j - \mathcal{A}_l) \tau_{(i)}\tau_{(jl)} + (\mathcal{A}_l - \mathcal{A}_i) \tau_{(j)}\tau_{(il)} + (\mathcal{A}_i - \mathcal{A}_j) 
\tau_{(l)}\tau_{(ij)} =0,
\qquad 1\leq i< j < l \leq \tilde{N},
\end{equation}
where $\mathcal{A}_i$ is an arbitrary function of the variable $n_i$, $i=1, \dots ,\widehat{N}$. Such a freedom in the Hirota system will be used to obtain generic non-isospectral form of its quasi-periodic reduction. 

\subsection{The non-commutative KP hierarchy}
To obtain the non-commutative analogue of the Kadomtsev--Petviashvilii (KP) hierarchy \cite{KNY-qKP,Dol-GD} we distinguish the last coordinate $k=n_{\widehat{N}}$, we put $N={\widehat{N}}-1$, and denote $n=(n_1,\dots , n_N)$. If we define
\begin{equation*}
\bphi(n,k) = \bphi_k(n), \qquad U_{{\widehat{N}},i}(n,k) = u_{i,k}(n),
\end{equation*}
then the linear system \eqref{eq:lin-dKP} gives the linear problem 
\begin{equation} \label{eq:lin-KP-phi}
\bphi_{k+1} - \bphi_{k(i)}  = \bphi_{k} u_{i,k} ,\qquad k\in\ZZ, 
 \qquad i=1,\dots ,N,
\end{equation}
and the potentials $u_{i,k}:\ZZ^{N}\to\DD$ satisfy the compatibility conditions
\begin{equation} \label{eq:KP-u}
u_{j,k}u_{i,k(j)}  = u_{i,k} u_{j,k(i)}  , \qquad
u_{i,k(j)}  + u_{j,k+1}  = 
u_{j,k(i)} + u_{i,k+1}, \qquad 1 \leq i\neq j \leq N.
\end{equation}
In consequence we obtain the transformation rule
\begin{equation} \label{eq:KP-u-solved}
u_{i,k(j)} = ( u_{i,k} - u_{j,k})^{-1} u_{i,k} ( u_{i,k+1} - u_{j,k+1}), \qquad i\neq j,
\end{equation} 
which can be written as a non-commutative discrete KP map, see Figure~\ref{fig:3D-GD-u},
\begin{equation*}
(\mathbf{u}_i, \mathbf{u}_j) \mapsto 
(\mathbf{u}_{i(j)}, \mathbf{u}_{j(i)} ) ,
\qquad \mathbf{u}_i = (u_{i,k}), \qquad k\in\ZZ. 
\end{equation*}
In recent studies on discrete integrable systems the property of multidimensional consistency \cite{ABS,FWN-cons} is considered as the main concept of the theory. Roughly speaking, it is the possibility of extending the number of independent variables of a given nonlinear system by adding its copies in different directions without creating this way inconsistency or multivaluedness.
For Desargues maps (in notation of the KP hierarchy) such a problem can occur in construction of the point $\bphi_{k(ijl)}$, which however (by the Desargues theorem \cite{Coxeter}) is the \emph{single} intersection point of three lines $\langle \bphi_{k(ij)} , \bphi_{k+1(ij)} \rangle$, $\langle \bphi_{k(il)} , \bphi_{k+1(il)} \rangle$, and $\langle \bphi_{k(jl)} , \bphi_{k+1(jl)} \rangle$,
see Figure~\ref{fig:Desargues-phi-k}. We would like to stress that this geometric theorem is the source of the multidimensional consistency of other systems described in the paper.

In particular, one can notice, that when we consider system \eqref{eq:KP-u} in three independent variables, there are two ways to obtain $\mathbf{u}_{i,(j l)}$ starting from 
$\mathbf{u}_i$, $\mathbf{u}_j$, $\mathbf{u}_l$ and using the map \eqref{eq:KP-u-solved}. 
\begin{Prop} \label{prop-3D-u-KP}
The KP map \eqref{eq:KP-u-solved} is three dimensionally consistent.
\end{Prop}
\begin{proof}
To have to show that two expressions below are equal (we assume that the indices $i$, $j$, $l$ are distinct)
\begin{equation*}
[u_{i,k(j)}]_{(l)} =   [( u_{i,k} - u_{j,k})^{-1} u_{i,k} ( u_{i,k+1} - u_{j,k+1}) ]_{(l)} , \qquad
[u_{i,k(l)}]_{(j)} =   [( u_{i,k} - u_{l,k})^{-1} u_{i,k} ( u_{i,k+1} - u_{l,k+1}) ]_{(j)} .
\end{equation*}
It is convenient to note first the following identity
\begin{equation} \label{eq:identity-3D-u}
( u_{i,k} - u_{j,k}) ( u_{i,k} - u_{l,k})_{(j)} = ( u_{i,k} - u_{l,k}) ( u_{i,k} - u_{j,k})_{(l)} ,
\end{equation}
which can be verified directly using the KP~map~\eqref{eq:KP-u-solved}. The rest are simple algebraic manipulations using twice the above identity (for $k$ and $k+1$).
\end{proof}
\begin{figure}
\includegraphics[width=12cm]{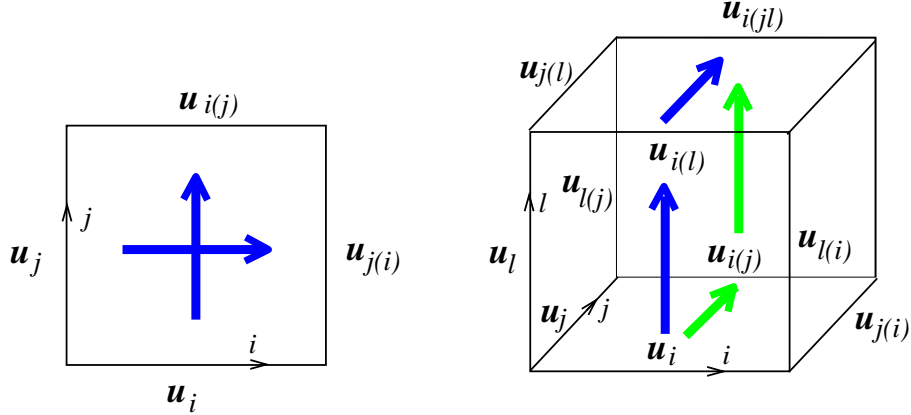}
\caption{The discrete KP map and its three dimensional consistency} 
\label{fig:3D-GD-u}
\end{figure}
We remark that because of infinite number of components of the fields $\mathbf{u}_i = (u_{i,k})_{k\in\ZZ}$ the above three-dimensional consistency of the KP~map~\eqref{eq:KP-u-solved} is equivalent to (local) four-dimensional consistency of the non-commutative Hirota system~\eqref{eq:comp-U}. Such a form is however convenient to demonstrate three-dimensional consistency of Gel'fand--Dikii lattice equations.

The fields $\mathbf{u}_{i}$ are attached to edges of the $\ZZ^N$ lattice. In \cite{Dol-GD} we studied the vertex form of the equations. Notice that the first part of equations \eqref{eq:KP-u} allows to define potentials $r_{k}$ such that
\begin{equation}
u_{i,k} = r_k^{-1} r_{k(i)},
\end{equation}
while the remaining part gives the map
\begin{equation} \label{eq:KP-r}
 r_{k(ij)}  = (r_{k(j)}^{-1} - r_{k(i)}^{-1})^{-1} r_{k+1}^{-1} (r_{k+1(i)} - r_{k+1(j)} )
\end{equation}
which allows to express the field in the fourth vertex of the elementary quadrilateral in terms of the fields in other three vertices. Again, the value of $\mathbf{r}_{(ijl)} = (r_{k(ijl)})_{k\in\ZZ}$ in the eighth vertex of the combinatorial cube (see Figure~\ref{fig:3D-GD}) can be calculated from initial values $\mathbf{r}$, $\mathbf{r}_{(i)}$, $\mathbf{r}_{(j)}$ and $\mathbf{r}_{(l)}$ in three different ways. It turns out that these three results are identical, i.e. the vertex KP map is multidimensionally consistent.  
\begin{figure}
\includegraphics[width=10cm]{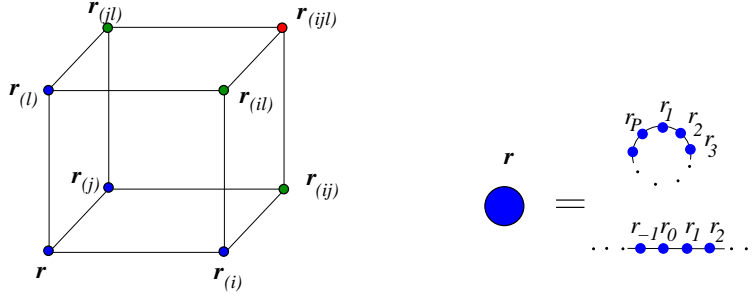}
\caption{Three dimensional consistency of the vertex KP map} 
\label{fig:3D-GD}
\end{figure}

\section{Quadrilateral lattice maps and their $B$ and $C$ reductions}
In this Section we recall \cite{Dol-Des,Dol-WCR-Hirota} the relation between Desargues maps and the multidimensional lattices of planar quadrilaterals \cite{DCN,MQL}. We concentrate on describing the relation between the Hirota system (called also discrete $A$-KP system) and its reductions to the discrete $B$-KP system (the Miwa system \cite{Miwa}) and the discrete $C$-KP system \cite{Kashaev-LMP,Schief-JNMP}. 

\subsection{Quadrilateral lattice maps}
Let $\widehat{N}=2K-1$ be odd, we split the standard basis vectors $(\mathbf{e}_i)_{i=1}^{2K}$ into $K$ pairs $( \mathbf{e}_1, \mathbf{e}_2)$, \dots , $( \mathbf{e}_{2K-1}, \mathbf{e}_{2K})$. 
We pick up the following root vectors $\mathbf{E}_i = \mathbf{e}_{2i-1} - \mathbf{e}_{2i}$, $i=1,\dots ,K$, of the root lattice $Q(A_{2K-1})$. These vectors satisfy the orthogonality relations
\begin{equation*}
( \mathbf{E}_i | \mathbf{E}_j ) = 2 \delta_{ij}, 
\end{equation*}
and generate the ${\mathbb Z}^K$ sub-lattice (with rescaled standard scalar product) in the root lattice $Q(A_{2K-1})$. 

It is not difficult to see that given Desargues map $\Phi:Q(A_{2K-1})\to{\mathbb P}^M({\mathbb D})$ then for arbitrary fixed
$\widehat{n}\in Q(A_{2K-1})$ the four points $\Phi(\widehat{n})$, $\Phi(\widehat{n}+\mathbf{E}_i)$, $\Phi(\widehat{n}+\mathbf{E}_j)$ and $\Phi(\widehat{n}+\mathbf{E}_i+\mathbf{E}_j)$ are coplanar (see Figure~\ref{fig:D-QL-ij}). In what follows we denote shifts in $\mathbf{E}_i$ by subscripts in square brackets, i.e. the four points above are given by $\Phi$, $\Phi_{[i]}$, $\Phi_{[j]}$, and $\Phi_{[ij]}$. To describe the relation in more detail
we introduce the following change of ${\mathbb Z}^{2K-1}$ coordinates in the lattice $Q(A_{2K-1})$
\begin{equation} \label{eq:n-m-l}
\widehat{n} = \sum_{i=1}^{2K-1} n_i \mathbf{\varepsilon}_i =
-\sum_{j=1}^K m_j \mathbf{E}_j + \sum_{j=1}^K \ell_j 
\mathbf{e}_{2j}, 
\end{equation}
i.e. $m_j = n_{2j-1}$, $\ell_j = - (n_{2j-1} + n_{2j})$, $1\leq j\leq K$, 
where we also defined $n_{2K} = -(n_1 + \ldots + n_{2K-1})$, which implies $\ell_1 + \ldots + \ell_K =0$. We have therefore $m= \sum_{j=1}^K m_j \mathbf{E}_j \in Q(B_K)={\mathbb Z}^K$, and $\ell = \sum_{j=1}^K \ell_j  \bar{\mathbf{e}}_i \in Q(A_{K-1})$, where $\bar{\mathbf{e}}_i = \mathbf{e}_{2i}$.
For fixed $\ell\in Q(A_{K-1})$ define the
map $\psi^\ell:{\mathbb Z}^K\to {\mathbb P}^M$ given by
$\psi^\ell(m) = \Phi(\widehat{n})$, where the relation between $n$ and $m$ and $\ell$ is
given above.

As described previously, the points $\psi^\ell$, $\psi^\ell_{[i]}$,  $\psi^\ell_{[j]}$, and $\psi^\ell_{[ij]}$ are coplanar. 
The transformation of the map $\psi^\ell\colon{\mathbb Z}^K\to {\mathbb P}^M$ into 
$\psi^{\ell+\bar{\mathbf{e}}_i -\bar{\mathbf{e}}_j}\colon {\mathbb Z}^K\to {\mathbb P}^M$, is called the Laplace transformation $\mathcal{L}_{ij}$ \cite{DCN,TQL}. Geometrically it is given by intersection of opposite tangent lines of planar quadrilaterals, as visualized in Figure~\ref{fig:D-QL-ij}. The Laplace transformations satisfy relations
\begin{equation*}
\mathcal{L}_{ij} \circ \mathcal{L}_{ji}  = \text{id} \; ,\qquad
\mathcal{L}_{jk} \circ \mathcal{L}_{ij} = \mathcal{L}_{ik}, \qquad
\mathcal{L}_{ki} \circ \mathcal{L}_{ij} = \mathcal{L}_{kj}, 
\end{equation*}
which express the root lattice $Q(A_{K-1})$ domain of definition of the variable $\ell$. We remark that the geometric construction of the Laplace transformations allows to recover from generic quadrilateral lattice map of $\ZZ^K$ the corresponding Desargues map of $Q(A_{2K-1})$.

\begin{figure}
\includegraphics[width=11cm]{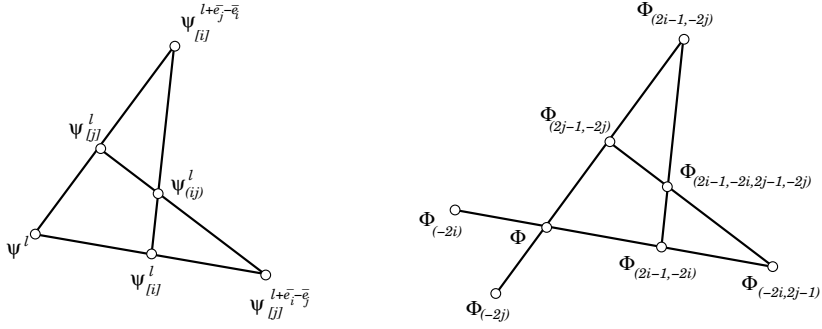}
\caption{Quadrilateral lattice maps and Desargues maps} 
\label{fig:D-QL-ij}
\end{figure}

\subsection{The $\tau$-function form of discrete Darboux equations}
Both shifts in $m\in Q(B_K)=\ZZ^K$ or $\ell\in Q(A_{K-1})$ variables involve double shifts in the original variable $\widehat{n}\in Q(A_{2K-1})$. Let us present the corresponding form~\cite{ABS-octahedron} of the Hirota system~\eqref{eq:H-M}
\begin{equation} \label{eq:H-M-sub}
\tau_{(ij)}\tau_{(kl)} - \tau_{(ik)}\tau_{(jl)} + \tau_{(il)}\tau_{(jk)} = 0, \qquad 1\leq i < j < k < l \leq 2K-1,
\end{equation}
which can be derived from three copies of equations~\eqref{eq:H-M} for triplets $(i,j,k)$, $(i,j,l)$ and $(i,k,l)$. 
\begin{Rem}
Notice that due to the structure of the $A$-type root lattice the six points involved in equation~\eqref{eq:H-M-sub} are vertices of an octahedron Delaunay polytope.
\end{Rem}
Then equation~\eqref{eq:H-M-sub} for indices $(2i-1,2i,2j-1,2j)$, $1\leq i < j < K$, shifted backwards in $(n_{2i},n_{2j})$ reads
\begin{equation} \label{eq:tau-Darboux-1}
\tau^{\ell} \tau^{\ell}_{[ij]} = \tau^{\ell}_{[i]}\tau^{\ell}_{[j]} +  
\tau^{\ell+\bar{\mathbf{e}}_j -\bar{\mathbf{e}}_i}_{[i]} 
\tau^{\ell+\bar{\mathbf{e}}_i -\bar{\mathbf{e}}_j}_{[j]},
\end{equation}
where we used also the "quadrilateral lattice notation" $\tau^{\ell}(m) \leftrightarrow \tau(n)$. By taking indices $(2i,2j,2k-1,2k)$, $1\leq i \neq j \neq k < K$, we obtain \cite{DMMMS,DS-sym,Dol-tau-QL}
\begin{equation} \label{eq:tau-Darboux-2}
\tau^{\ell} \tau^{\ell+\bar{\mathbf{e}}_i -\bar{\mathbf{e}}_j}_{[k]} = 
\tau^{\ell}_{[k]}  \tau^{\ell+\bar{\mathbf{e}}_i -\bar{\mathbf{e}}_j} + \mathrm{sgn}(j-i) \mathrm{sgn}(k-j)\mathrm{sgn}(i-k)
\tau^{\ell+\bar{\mathbf{e}}_i -\bar{\mathbf{e}}_k}_{[k]} \tau^{\ell+\bar{\mathbf{e}}_k -\bar{\mathbf{e}}_j},
\end{equation}
where in the above we presented the form of equations valid for arbitrary ordering of indices $i,j,k$. Little more care should be taken when one of the indices is $K$, but the final result is the same. For example to obtain equation \eqref{eq:tau-Darboux-1}
with $j=K$ we make use of the original Hirota equation~\eqref{eq:H-M} for the triplet $(2i-1,2i,2K-1)$. 

Equations \eqref{eq:tau-Darboux-2} can be rewritten in terms the so called rotation coefficients
\begin{equation}
\beta_{ij}^{\ell} = \mathrm{sgn}(j-i) \left( \frac{\tau^{\ell+\bar{\mathbf{e}}_i -\bar{\mathbf{e}}_j}}{\tau^{\ell}} \right)_{[j]}, \qquad i\neq j,
\end{equation}
as the discrete Darboux equations~\cite{BoKo,MQL} (notice a shift in definition of the rotation coefficients)
\begin{equation} \label{eq:d-Darboux}
\beta_{ij[k]}^{\ell} = \beta_{ij}^{\ell} + \beta_{ik[j]}^{\ell} \beta_{kj}^{\ell} , \qquad i,j,k \quad \text{distinct}. 
\end{equation}
The natural continuous limit of equations \eqref{eq:d-Darboux}, which in modern literature~\cite{KvL} are called also the $K$-wave equations, plays a fundamental role in the theory of conjugate nets~\cite{Darboux-OS}.

We remark that the discrete Darboux equations involve rotation coefficients with fixed value of the Laplace transformation index $\ell\in Q(A_{K-1})$. In what follows we present reductions of quadrilateral lattices maps which lead to equations with fixed $\ell$ on the $\tau$-function level.

\subsection{The $B$-KP and $C$-KP reductions}
Below we present a reduction from the Hirota system to the discrete $B$-KP and $C$-KP systems. Actually, we perform the reduction on the level of the $\tau$-function form \eqref{eq:tau-Darboux-1}-\eqref{eq:tau-Darboux-2} of the discrete Darboux equations, i.e. we already have changed variables according to equation \eqref{eq:n-m-l} and fixed the Laplace transformation variable $\ell_0 \in Q(A_{K-1})$. Then we give a constraint relating the Laplace transformed $\tau$-functions 
$\tau^{\ell_0+\bar{\mathbf{e}}_i -\bar{\mathbf{e}}_j}$ and $\tau^{\ell_0}$. Up to author's knowledge such a discrete $B$-KP constraint was not written down before. We remark that our transition from $2K-1$ Desargues map variables to $K$ variables of the quadrilateral lattice maps corresponds to classical results~\cite{JimboMiwa} where in reduction from KP hierarchy to the $B$- and $C$-KP hierarchies "approximately half" of the original KP times is put to zero.

In \cite{CQL} it was considered an integrable reduction of discrete Darboux equations imposed by the constraint
\begin{equation} \label{eq:CKP-constraint}
\beta_{ij}^{\ell_C}\beta_{jk}^{\ell_C}\beta_{ki}^{\ell_C} = \beta_{ji}^{\ell_C}\beta_{kj}^{\ell_C}\beta_{ik}^{\ell_C} , 
\qquad i,j,k \quad \text{distinct},
\end{equation}  
for certain $\ell = \ell_C$. It was also shown that using the allowed gauge freedom in definition of the $\tau$-function one can bring the constraint to the form (in the present notation)
\begin{equation} \label{eq:CKP-constraint-tau}
\tau^{\ell_C+\bar{\mathbf{e}}_i -\bar{\mathbf{e}}_j}_{[j]} + \tau^{\ell_C+\bar{\mathbf{e}}_j -\bar{\mathbf{e}}_i}_{[i]} = 0,
\qquad i\neq j, 
\end{equation}
which, due to equations~\eqref{eq:tau-Darboux-1}, allowed in \cite{Schief-JNMP} to rewrite the discrete Darboux equations \eqref{eq:d-Darboux} in terms of single $\tau$-function $\tau = \tau^{\ell_C}$ in the form
\begin{equation} \label{eq:CKP} 
\left(\tau_{[i]} \tau_{[jk]} - \tau_{[j]} \tau_{[ik]} +
\tau_{[k]} \tau_{[ij]}  - \tau \tau_{[ijk]} \right)^2 - 
4 \left( \tau_{[i]} \tau_{[jk]} \tau_{[k]} \tau_{[ij]} +
\tau_{[j]} \tau_{[ik]} \tau \tau_{[ijk]} \right) =
4 \tau_{[i]} \tau_{[j]} \tau_{[k]}   \tau_{[ijk]} +
4 \tau \tau_{[ij]} \tau_{[jk]} \tau_{[ik]} ,
\end{equation}
identified there as the superposition principle of $\tau$-functions of  the $C$-KP hierarchy. We mention that equation \eqref{eq:CKP} was obtained earlier in \cite{Kashaev-LMP} on a different basis. Although if the form of equation \eqref{eq:CKP} does not look symmetric, it is invariant with respect to permutations of indices. However, from equation \eqref{eq:CKP} the function $\tau_{[ijk]}$ can be calculated as a root of the second degree polynomial equation, which introduces sign ambiguites and creates potential problems for four dimensional consistency of the equation. This problem was studied in \cite{Atkinson}, and it was solved by re-introducing (essentially) the functions $\tau^{\ell_C+\bar{\mathbf{e}}_i -\bar{\mathbf{e}}_j}_{[j]}$, $i<j$, together with corresponding additional equations.

The discrete $B$-KP system
\begin{equation} \label{eq:BKP-M}
\mu \mu_{[ijk]} = \mu_{[i]} \mu_{[jk]} - \mu_{[j]} \mu_{[ik]} + \mu_{[k]} \mu_{[ij]} , \qquad i<j<k,
\end{equation}
known also as the Miwa equations~\cite{Miwa}, was studied in \cite{BQL} in relation to the so called
$B$-quadrilateral lattice maps. Our goal is to describe the reduction on the $\tau$-function level. From the known relation between the $\tau$-functions of the KP and $B$-KP hierarchies \cite{DKJM,Hirota-book}, one can expect that $\mu^2 = \tau^{\ell_B}$ for certain $\ell_B\in Q(A_{K-1})$. Under such assumption the Miwa system \eqref{eq:BKP-M} implies the following equations, which we write down in the form similar to that of~\eqref{eq:CKP} for $\tau = \tau^{\ell_B}$
\begin{equation} \label{eq:BKP-tau} 
\left[ \left(\tau_{[i]} \tau_{[jk]} - \tau_{[j]} \tau_{[ik]} +
\tau_{[k]} \tau_{[ij]}  - \tau \tau_{[ijk]} \right)^2 - 
4 \left( \tau_{[i]} \tau_{[jk]} \tau_{[k]} \tau_{[ij]} +
\tau_{[j]} \tau_{[ik]} \tau \tau_{[ijk]} \right) \right]^2 =
64 \tau \tau_{[i]} \tau_{[j]} \tau_{[k]} 
\tau_{[ij]} \tau_{[jk]} \tau_{[ik]}   \tau_{[ijk]} ,
\end{equation}
Notice however, that equation \eqref{eq:BKP-tau} is a polynomial of the fourth degree in $\tau_{[ijk]}$, what introduces even more sign ambiguites then in the case of the discrete $C$-KP equation~\eqref{eq:CKP}.

Basing on our previous results~\cite{BQL} on geometric interpretation of the Miwa system we propose a constraint, analogous to~\eqref{eq:CKP-constraint-tau}.
\begin{Prop} \label{prop:BKP}
Given solution $\tau \colon \ZZ^{2K-1} \to \FF$ of the Hirota system \eqref{eq:H-M} such that after transformation to quadrilateral lattice variables \eqref{eq:n-m-l} for certain $\ell_B\in Q(A_{K-1})$ we have
\begin{equation} \label{eq:BKP-constraint-tau}
\left( \tau^{\ell_B+\bar{\mathbf{e}}_i -\bar{\mathbf{e}}_j}_{[j]} - \tau^{\ell_B+\bar{\mathbf{e}}_j -\bar{\mathbf{e}}_i}_{[i]} \right)^2= 4 \tau^{\ell_B}_{[i]} \tau^{\ell_B}_{[j]},
\qquad i\neq j, 
\end{equation}
then:
\begin{enumerate}
\item The function $\tau = \tau^{\ell_B}$  satisfies equation \eqref{eq:BKP-tau}.
\item One can consistently parametrize the constraint \eqref{eq:BKP-constraint-tau} in terms of a
function $\mu \colon \ZZ^K \to \FF$ such that $\mu^2 = \tau^{\ell_B}$, and for $i<j$
\begin{equation*}
\tau^{\ell_B+\bar{\mathbf{e}}_i -\bar{\mathbf{e}}_j}_{[j]}  = - (-1)^{\sum_{i\leq k<j}m_k}
\left(  \mu \mu_{[ij]} + \mu_{[i]} \mu_{[j]} \right) , \qquad 
\tau^{\ell_B+\bar{\mathbf{e}}_j -\bar{\mathbf{e}}_i}_{[i]}   = - (-1)^{\sum_{i\leq k<j}m_k}
\left(  \mu \mu_{[ij]} - \mu_{[i]} \mu_{[j]} \right) .
\end{equation*}
\item Under such parametrization the Darboux equations reduce to the Miwa system of equations~\eqref{eq:BKP-M} for $\mu$. 
\end{enumerate}
\end{Prop} 
\begin{proof}
The first point can be checked directly from equations~\eqref{eq:tau-Darboux-1}-\eqref{eq:tau-Darboux-2} supplemented by the constraint~\eqref{eq:BKP-constraint-tau}. It is convenient to note two other equivalent forms of the constraint
\begin{equation*}
\left( \tau^{\ell_B+\bar{\mathbf{e}}_i +\bar{\mathbf{e}_j}}_{[j]} + \tau^{\ell_B+\bar{\mathbf{e}}_j -\bar{\mathbf{e}}_i}_{[i]} \right)^2= 4 \tau^{\ell_B} \tau^{\ell_B}_{[ij]} \qquad \text{or} \qquad 
\left( \tau^{\ell_B+\bar{\mathbf{e}}_i +\bar{\mathbf{e}}_j}_{[j]} \right)^2 + \left( \tau^{\ell_B+\bar{\mathbf{e}}_j -\bar{\mathbf{e}}_i}_{[i]} \right)^2 = 2 \left( \tau^{\ell_B}_{[i]} \tau^{\ell_B}_{[j]} +  \tau^{\ell_B} \tau^{\ell_B}_{[ij]} \right) .
\end{equation*}
To demonstrate the second part we notice that equations \eqref{eq:BKP-constraint-tau} and \eqref{eq:tau-Darboux-1} imply that 
$\tau^{\ell_B+\bar{\mathbf{e}}_i +\bar{\mathbf{e}}_j}_{[j]}$ and $\tau^{\ell_B+\bar{\mathbf{e}}_j -\bar{\mathbf{e}}_i}_{[i]} $ are the roots $x_{1,2}$ of the following second degree equation (with the sign ambiguity)
\begin{equation*}
x^2 \pm  2 x \sqrt{\tau^{\ell_B} \tau^{\ell_B}_{[ij]} } + \tau^{\ell_B} \tau^{\ell_B}_{[ij]} -  \tau^{\ell_B}_{[i]} \tau^{\ell_B}_{[j]} =0, \qquad \text{i.e.} \qquad
x_{1,2} = \pm \left( \sqrt{\tau^{\ell_B} \tau^{\ell_B}_{[ij]} } \pm \sqrt{\tau^{\ell_B}_{[i]} \tau^{\ell_B}} \right) ,
\end{equation*}
which give desired form of the function $\mu$, where we fix the signs following~\cite{BQL}. Finally, for arbitrary ordering of indices $i,j,k$, equations \eqref{eq:tau-Darboux-2} shifted in $m_j$ variable give directly the Miwa system~\eqref{eq:BKP-M}. 
\end{proof}

As a consequence of the Miwa system \eqref{eq:BKP-M} on can derive its double-shift version
\begin{equation} \label{eq:DKP}
\mu \mu_{[ijkl]} - \mu_{[ij]} \mu_{[kl]} + \mu_{[ik]} \mu_{[jl]} - \mu_{[ik]} \mu_{[jl]} = 0, \qquad 1\leq i<j<k<l \leq K,
\end{equation}
obtained in  \cite{Bobenko-IDS} while studying four dimensional consistency of \eqref{eq:BKP-M}. We discuss this equation in the spirit of integrable systems on root lattices, where the Hirota equations (discrete $A$-KP system) are defined on root lattices of type $A$, the discrete $C$-KP \eqref{eq:CKP} system and the discrete $B$-KP system in its Miwa form \eqref{eq:BKP-M} or in the form \eqref{eq:BKP-tau} presented here are defined on root lattices of type $B$. Then equation \eqref{eq:DKP} involves independent variables from the root lattice $Q(D_K)$, which is obtained from $\ZZ^K=Q(B_K)$ by first colouring its points alternately white and black, and then taking the black points (checkerboard lattice). Some integrable systems on such sub-lattice obtained as reductions of discrete $B$-KP equation or its (discrete Moutard) linear problem were considered in~\cite{Dol-Gr-Nie-San, Dol-Nie-San-tr-hon, San-Dol-Nie-Toda-sq-tr}. We also mention that both constraints \eqref{eq:CKP-constraint-tau} and \eqref{eq:BKP-constraint-tau} impose conditions on Laplace transforms of the reduced quadrilateral lattice maps; this point of view was analysed 
in~\cite{Nieszporski-Laplace}.

\section{Periodic reduction of Desargues maps}
\label{sec:per-red}
In this Section we return to non-commuting dependent variables. Our goal is to derive integrable systems of lower dimension. We start from imposing periodicity constraint 
$\Phi(n,k+P) = \Phi(n,k)$ on the level of Desargues maps, and study its implications on the level of the functions 
$\bphi_k$ and $u_{i,k}$. 
\subsection{Discrete non-isospectral modified Gel'fand--Dikii systems}
On the level of the homogeneous coordinates we must have functions (the monodromy factors) $\lambda_k\colon \ZZ^N\to\DD^\times$ such that, $\bphi_{k+P} = \bphi_k \lambda_k$. Since the factors do not change the structure of the linear problem \eqref{eq:lin-KP-phi} then they satisfy the constraint $\lambda_{k+1} = \lambda_{k(i)}$, $i=1,\dots ,N$.
The corresponding transformation of the potentials $u_{i,k}$ is given by
\begin{equation}
u_{i,k+P} = \lambda_k^{-1} u_{i,k}\lambda_{k(i)}.
\end{equation}

The periodicity condition implies then the following linear system
\begin{equation} \label{eq:Lm-kp-K}
\left(\bphi_1, \bphi_2, \dots , \bphi_P \right)_{(i)}= 
\left(\bphi_1, \bphi_2, \dots , \bphi_P \right)
\left(  \begin{array}{ccccc}  
-u_{i,1} & 0  & \cdots & 0 & \lambda_1 \\
1 & -u_{i,2}  & 0 & \hdots  & 0 \\
0 & 1 & \ddots &  &  \vdots \\
\vdots &  & & -u_{i,P-1}  & 0 \\
0 & 0 & \ \hdots  & 1 & -u_{i,P}  \end{array} \right) ,
\end{equation}
where $\lambda_1$ is a function of the variable $n_\sigma = n_1 + \dots + n_N$ and plays the role of variable spectral parameter. The corresponding reduction of the KP map takes the form
\begin{align*} 
u_{i,k(j)} & = ( u_{i,k} - u_{j,k})^{-1} u_{i,k} ( u_{i,k+1} - u_{j,k+1}), \qquad k=1, \dots , P-1, \qquad i\neq j,\\
u_{i,P(j)} & = ( u_{i,P} - u_{j,P})^{-1} u_{i,P} \lambda_1^{-1}( u_{i,1} - u_{j,k+1}) \lambda_{1(\sigma)}.
\end{align*} 
and is three-dimensionally consistent.

The periodic reduction of Desargues maps results on the level of the fields $r_k$ as the constraint $r_{k+P} = r_k \lambda_k$, and gives the multidimensionally consistent map 
\begin{align} \label{eq:GD-r-mu}
 r_{k(ij)} & = (r_{k(j)}^{-1} - r_{k(i)}^{-1})^{-1} r_{k+1}^{-1} (r_{k+1(i)} - r_{k+1(j)} ), \qquad k=1, \dots , P-1,\\ 
 r_{P(ij)} & = (r_{P(j)}^{-1} - r_{P(i)}^{-1})^{-1} r_{1}^{-1} \lambda_1^{-1} (r_{1(i)} - r_{1(j)} ) \lambda_{1(\sigma)}, \qquad \qquad i\neq j .
\end{align}

\subsection{Commutative specialization and simplest reductions}
Below we show how in the commutative case 
there show up functions of single variables, the presence of which is indispensable in making further reductions to Painlev\'{e} type dynamical systems. 
We perform the discussion on the level of the edge functions $u_{i,k}$, while in \cite{Dol-GD} we worked with the vertex fields $r_{k}$.

Define $\mathcal{U}_{i,k}= u_{i,k} u_{i,k+1} \dots u_{i,k+P-1} \lambda_k^{-1}$, then it is easy to check that
\begin{equation}
\mathcal{U}_{i,k+1} = \mathcal{U}_{i,k}, \qquad \mathcal{U}_{i,k(j)} = \mathcal{U}_{i,k}, \qquad j\neq i,
\end{equation}
i.e. $\mathcal{U}_{i,k}$ does not depend on the index $k$ (which we skip from now on) and it is a function of the single variable $n_i$. 
On the level of the vertex functions $r_k$ we can then define 
$\mathcal{R}_k = r_k r_{k+1} \dots r_{k+P-1} \mathcal{M}_k^{-1}$, where $\mathcal{M}_{k+1} = \lambda_k \mathcal{M}_k$. Then again $\mathcal{R}_k$ is independent of $k$ (which we skip from now on) and factorizes into a product of functions of single arguments such that $\mathcal{R}_{(i)} = \mathcal{U}_{i} \mathcal{R}$. Using these facts we replace the system \eqref{eq:GD-r-mu} by a new one involving $P-1$ unknown functions (we may chose $r_1, \dots , r_{P-1}$ and certain numbers of functions of single variables.

Let us present an example \cite{Dol-GD} of the non-isospectral non-autonomous lattice modified KdV system which can be obtained in the simplest case $P=2$. To match with known form of its isospectral version we introduce functions 
$\mathcal{G} = (\mathcal{R})^{1/P}$, $\mathcal{F}_i = (\mathcal{U}_i)^{1/P}$, and we express $r_1$ and $r_2$ in terms of single field $x$ as follows 
\begin{equation*}
r_1 = x\mathcal{G} , \qquad r_2 = \frac{\mathcal{G} \mathcal{M}_1}{x} , 
\end{equation*}
then equations \eqref{eq:GD-r-mu} reduce to  
\begin{equation} \label{eq:l-mKdV-ni} 
x_{(ij)} = \lambda_1 \, x \;\frac{x_{(i)}\mathcal{F}_j - x_{(j)}\mathcal{F}_i}{x_{(j)}\mathcal{F}_j - x_{(i)}\mathcal{F}_i} , \qquad i\neq j.
\end{equation}

Similar procedure in the periodic reduction with $P=3$ gives a two-component system, which after elimination of one field gives rise \cite{Dol-GD} to the following lattice equation, which is non-isospectral and non-autonomous version of the lattice modified Boussinesq equation \cite{FWN-lB}
\begin{equation*}
\left( \frac{ \lambda_1 }{  y_{(ij)} }  
\left( y_{(i)} \mathcal{F}_j - y_{(j)} \mathcal{F}_i \right) \right)_{(ij)} 
- \lambda_1 y 
\left( \frac{ \mathcal{F}_j }{ y_{(j)} } -  \frac{ \mathcal{F}_i }{ y_{(i)} } \right) 
= 
\left(  \frac{ y_{(ij)} }{ \lambda_1 y } \, 
\frac{ y_{(j)} \mathcal{F}_j^2 -y_{(i)} \mathcal{F}_i^2 }{ y_{(i)} \mathcal{F}_j - y_{(j)} \mathcal{F}_i } \right)_{(j)}  - 
\left(  \frac{ y_{(ij)} }{ \lambda_1 y } 
\frac{ y_{(i)} \mathcal{F}_i^2 -y_{(j)} \mathcal{F}_j^2 }{ y_{(j)} \mathcal{F}_i - y_{(i)} \mathcal{F}_j } \right)_{(i)}  , \quad i\neq j.
\end{equation*}

Equations of the non-isospectral and non-autonomous modified lattice Gel'fand--Dikii hierarchy can be also recovered on the $\tau$-functions level, but we have to start from the non-autonomous Hirota system~\eqref{eq:H-M-na} with the distinguished last variable. In the specification
\begin{equation*}
\mathcal{A}_i = \mathcal{F}_i + 1, \qquad 1\leq i \leq N, \qquad \mathcal{A}_{N+1} =1,
\end{equation*}
we obtain first the $\tau$-function formulation of the (non-autonomous and commutative) KP hierarchy
\begin{equation}
\mathcal{F}_j\tau_{k(i)} \tau_{k+1 (j)} - \mathcal{F}_i\tau_{k(j)} \tau_{k+1 (i)} + (\mathcal{F}_i - \mathcal{F}_j) \tau_{k+1} \tau_{k(ij)} = 0, \qquad i\neq j.
\end{equation}
Its consequence
\begin{equation}
\left( \left( \frac{\tau_{k}}{\tau_{k+1}}\right)_{(j)} \mathcal{F}_i - \left( \frac{\tau_{k}}{\tau_{k+1}}\right)_{(i)} \mathcal{F}_j \right)  \left( \frac{\tau_{k+1}}{\tau_{k}}\right)_{(ij)} =
\frac{\tau_{k+1}}{\tau_{k+2}}
\left( \left( \frac{\tau_{k+2}}{\tau_{k+1}}\right)_{(i)} \mathcal{F}_i - 
\left( \frac{\tau_{k+2}}{\tau_{k+1}}\right)_{(j)} \mathcal{F}_j \right) ,
\end{equation}
after identification 
\begin{equation*}
r_k = \frac{\tau_{k+1}}{\tau_{k}}\mathcal{G}, \qquad \mathcal{G}_{(i)} = \mathcal{F}_i \mathcal{G},
\end{equation*}
gives equations \eqref{eq:KP-r}. By imposing (quasi)-periodicty condition 
\begin{equation}
\tau_{k+P} = \tau_k \mathcal{M}_k, \qquad \mathcal{M}_{k+1} = \mathcal{M}_{k(i)} = \lambda_k \mathcal{M}_k,
\end{equation} 
we obtain (commutative version of) equations~\eqref{eq:GD-r-mu}.

\subsection{Reduction to $q$-Painleve equation of type $A_2+A_1$}
In the last Section we conclude our discussion of Desargues maps and the Hirota system, and of their various reductions by presenting the final reduction to the second Painlev\'{e} equation. According to Kruskal~\cite{GrammaticosRamani-review} 
Painlev\'{e} equations are located on a borderline between trivial integrability (usually linearisability) and non-integrability. They posses very interesting peculiar properties and have found numerous applications, see for example collection of articles in~\cite{Painleve-Conte}. It is well known that all the six Painlev\'{e} equations can be obtained as reductions of partial differential equations~\cite{AblowitzSegur,Adler,CGM}. However, in spite of various successful attempts
 \cite{FWN-lB,GRSWC,KNY-qKP,NRGO,RamaniGrammaticosJoshi} including also the most recent ones \cite{Ormerod,OKQ}, still there does not exist such a procedure for all the discrete Painlev\'{e} equations as classified in~\cite{Sakai}. In such a reduction procedure non-isospectral versions of integrable systems may bring additional parameters~\cite{LeviRagniscoRodriguez}.

Consider the non-isospectral non-autonomous modified lattice KdV system \eqref{eq:l-mKdV-ni}, which we write in the form
\begin{equation} \label{eq:lmKdV-ij}
x_{(j)} x_{(ij)} \mathcal{F}_j + \lambda x x_{(j)}\mathcal{F}_i =  x_{(i)} x_{(ij)} \mathcal{F}_i + \lambda x x_{(i)}\mathcal{F}_j, \qquad i\neq j,
\end{equation}
where we put $\lambda = \lambda_1$. 
Following \cite{KNY-qKP,RamaniGrammaticosJoshi} let us impose additional self-similarity constraint 
$x_{(1,2,\dots, N)}= \gamma x$, where $\gamma$ is a function to be determined from the consistency of the reduction with the lattice equation itself. After shifting equation \eqref{eq:lmKdV-ij} in all variables we 
 reproduce the original equation  \eqref{eq:lmKdV-ij} if 
\begin{enumerate}
\item $\gamma$ is a function of a single variable $n_\sigma = n_1 + n_2 + \dots + n_N$, and of period two,
\item $\lambda$ is a function of period $N$,
\item functions $\mathcal{F}_i$ are of the form $\mathcal{F}_i(n_i) = a_i q^{n_i}$ with non-zero constants $a_i$, $i=1,2,\dots ,N$, and $q$ common for all $i$.
\end{enumerate}
We will describe in detail the case $N=3$.
Equation \eqref{eq:lmKdV-ij} in variables $(n_1,n_2)$ shifted in $n_3$ gives, under the similarity constraint,  
\begin{equation}
\gamma_{(\sigma)} \frac{x_{(-1)}}{x_{(13)} } = 
\frac{ \lambda_{(\sigma)} x_{(3)} {\mathcal{F}}_2 + \gamma x {\mathcal{F}}_1 }
{{\gamma x \mathcal{F}}_2 + \lambda_{(\sigma)} x_{(3)} u {\mathcal{F}}_1}.
\end{equation}
Similarly, equation \eqref{eq:lmKdV-ij} in variables $(n_1,n_3)$ leads to
\begin{equation}
\frac{x_{(3)}}{x_{(1)} }  = 
\frac{ \lambda x {\mathcal{F}}_3 + x_{(13)}{\mathcal{F}}_1 }{x_{(13)}{\mathcal{F}}_3 + \lambda x {\mathcal{F}}_1}.
\end{equation}

We take $n_1$ as independent variable denoting by upper (lower) tilde $\tilde{}\,$ the corresponding forward (backward) shift, and we treat other variables as symmetry parameters. Define the "time function" $t$ and two dependent variables $f$ and $g$  by
\begin{equation*}
t = \frac{\mathcal{F}_3}{\mathcal{F}_1} , \quad \text{i.e.} \quad \tilde{t} = \frac{1}{q} t,
\qquad
f =  \frac{\lambda x}{x_{(13)} }, \qquad g =  \frac{\tilde{\lambda} x_{(3)}}{\gamma x},
\end{equation*}
which gives the standard form of the asymmetric $q-P_{II}$ equation equation with the symmetry group being the extended affine Weyl group of type $A_2+ A_1$
\begin{equation} \label{eq:q-P-III}
g \tilde{g}  = \frac{b}{f} \frac{(1 + f t)}{(t + f)}, \qquad 
f \underaccent{\tilde}{f}  = \frac{b}{g} \frac{(1 + gct)}{(ct + g)},
\end{equation}
where the parameters $b$ and $c$ read
\begin{equation*}
b = \frac{\lambda \tilde{\lambda} \tilde{\tilde{\lambda}}}{\gamma \tilde{\gamma}} , \qquad c = \frac{\mathcal{F}_2}{\mathcal{F}_3}.
\end{equation*}
Notice that the parameter $b$ is here an invariant combination of two periodic functions, which results in the previously known equation, i.e. the presence of the non-isospectrality not always gives something new.  

It is known that to obtain symmetric form of $q-P_{II}$ equation we consider the doubled lattice $\ZZ \cup (\ZZ + \frac{1}{2})$ with the corresponding half increment (denoted by hat $\hat{~}$) such that $\hat{t} = \frac{1}{\sqrt{q}} t$. Define the field $X$ on the doubled lattice by $X(n_1) = f(n_1)$, and $X(n_1 - \frac{1}{2}) = g(n_1)$. If we impose the reduction condition
$c= \sqrt{q}$ then we obtain single equation on the doubled lattice
\begin{equation}
\hat{X} \underaccent{\hat}{X} = b \frac{1 + tX}{X(t+X)},
\end{equation}
which is the standard symmetric form of $q-P_{II}$ equation. In the limit of small $\delta$, when we assume
\begin{equation}
b = e^{\alpha\delta^3}, \qquad q = e^{\delta^3/2}, \qquad t = -2 e^{ - z \delta^2/4 - \alpha\delta^3/4},
\qquad X = (1 + \delta w) e^{ - z \delta^2/4 + \alpha\delta^3/4},
\end{equation}
it gives the second Painlev\'e equation $P_{II}$
\begin{equation*}
w^{\prime\prime}(z) = 2w(z)^3 + z w(z) + \alpha.
\end{equation*}
 
Notice that originally equation \eqref{eq:q-P-III} was obtained in \cite{KruskalTamizhmaniGrammaticosRamani} as asymmetric generalization of the $q-P_{II}$ equation~\cite{RamaniGrammaticos}, and it admits a continuous limit to the Painlev\'{e}~III equation. For its interesting properties and special function solutions we refer to~\cite{KajiwaraKimura}.

%%%%%%%%%%%%%%%%%%%%%%%%%%%%%%%%%%%%%%%%%%%%%%%%
%% BACKMATTER
%%%%%%%%%%%%%%%%%%%%%%%%%%%%%%%%%%%%%%%%%%%%%%%%

\begin{theacknowledgments}
  The research was supported in part by Polish Ministry of Science and Higher Education grant No.~N~N202~174739. Author would like to thank the organizers of the 2nd International Workshop on Nonlinear and Modern Mathematical Physics for invitation and support.
\end{theacknowledgments}

\bibliographystyle{aipproc}   % if natbib is available

\end{document}